\documentclass[11pt, journal, onecolumn, a4paper]{IEEEtran}
\usepackage{cite}
\usepackage{enumitem}

\usepackage{amsfonts, amsmath, amssymb, amsthm}
\usepackage{array}
\usepackage{tikz, color, soul}
\usepackage{caption}
\usepackage{pgfplots}
\usepackage[hidelinks]{hyperref} 
\usetikzlibrary{plotmarks}
\usepackage{mathtools}
\usepackage[ruled,linesnumbered]{algorithm2e}
\usetikzlibrary{shapes.geometric,arrows,positioning}
\tikzstyle{decision} = [diamond, draw, fill=blue!20, 
    text width=5em, text badly centered, node distance=4cm, inner sep=0pt]
\tikzstyle{block} = [rectangle, draw, fill=blue!20,  text centered, rounded corners, minimum height=4em]
\tikzstyle{line} = [draw, -latex']
\tikzstyle{cloud} = [draw, ellipse,fill=red!20, node distance=6.6cm,
    minimum height=1em]
\tikzstyle{algorithm} = [rectangle, draw, fill=green!20,  text centered, rounded corners, minimum height=4em, minimum width =6em]
\tikzstyle{initialization} = [rectangle, draw,   text centered, minimum height=4em, minimum width =6em]

\usetikzlibrary{spy}
\usetikzlibrary{backgrounds}
\usetikzlibrary{decorations}
\usetikzlibrary{patterns}
\usetikzlibrary{arrows,matrix,positioning,calc}

\newtheorem{theorem}{Theorem}
\newtheorem*{theorem*}{Theorem}
\newtheorem{proposition}{Proposition}
\newtheorem{lemma}{Lemma}

\theoremstyle{definition}
\newtheorem{definition}{Definition}
\theoremstyle{remark}
\newtheorem*{remark}{Remark}

\newcommand{\M}{{\mathcal M}}

\usepackage{xcolor}
\newif\ifcomment
\commenttrue
\renewcommand{\epsilon}{\varepsilon}

\tikzstyle{block}=[draw, rectangle, minimum height=1cm, text width=1.5cm, text centered, draw=darkgray, font=\small]
\tikzstyle{block_medium}=[draw, rectangle, minimum height=1.5cm, text width=2cm, text centered, draw=darkgray, font=\small]
\tikzstyle{block_large}=[draw, rectangle, minimum height=2cm, text width=2cm, text centered, draw=darkgray, font=\small]
\tikzstyle{line} = [draw, -latex]

\usepackage{mathtools}
\newcommand{\defeq}{\vcentcolon=}
\newcommand{\eqdef}{=\vcentcolon}
\newcommand{\np}[1]{\ifcomment \textcolor{blue}{Nikita: #1} \fi}
\newcommand{\gm}[1]{\ifcomment \textcolor{red}{Georg: #1} \fi}

\def\qed{\hfill $\blacksquare$}
\pgfplotsset{compat=1.14}
\begin{document}
	\title{Coding with Noiseless Feedback over the Z-channel
	}
\author{%
    \IEEEauthorblockN{Christian Deppe, Vladimir Lebedev, Georg Maringer, and Nikita~Polyanskii}\\
	\thanks{
      The results of this paper  have partially been presented at the  26th International Computing and Combinatorics Conference (COCOON) 2020~\cite{deppe2020coding}.

     Christian Deppe was supported by the Bundesministerium f\"ur Bildung und Forschung (BMBF) through Grant 16KIS1005. Vladimir Lebedev's work was supported by the Russian Foundation for Basic Research (RFBR) under Grant No. 19-01-00364 and by RFBR and JSPS under Grant No. 20-51-50007. Georg Maringer's work was supported by the German Research Foundation (Deutsche Forschungsgemeinschaft, DFG) under Grant No.~WA3907/4-1. N. Polyanskii's work was supported by the German Research Foundation (Deutsche Forschungsgemeinschaft, DFG) under Grant No. WA3907/1-1.

      C.~Deppe and G.~Maringer are with the Institute for Communications Engineering, Technical University of Munich, Germany.  N.~Polyanskii is with the Institute for Communications Engineering, Technical University of Munich, Germany, and the  Center for Computational and Data-Intensive Science and Engineering, Skolkovo Institute of Science and Technology, Russia. V.~Lebedev is with the  Institute for Information Transmission Problems, Russian Academy of Sciences, Russia.

Emails: christian.deppe@tum.de, lebedev37@mail.ru, georg.maringer@tum.de, nikita.polyansky@gmail.com
    }}

  \IEEEoverridecommandlockouts
  \maketitle	

\begin{abstract}
In this paper, we consider encoding strategies for the Z-channel with noiseless feedback. We analyze the combinatorial setting where the maximum number of errors inflicted by an adversary is proportional to the number of transmissions, which goes to infinity. 
Without feedback, it is known that the rate of optimal asymmetric-error-correcting codes for 
the error fraction $\tau\ge 1/4$ vanishes as the blocklength grows. 
In this paper, we give an efficient feedback encoding scheme with $n$ transmissions that achieves a positive rate for any fraction of errors $\tau<1$ and  $n\to\infty$. Additionally, we state an upper bound on the rate of asymptotically long feedback asymmetric error-correcting codes. 
\end{abstract}

\section{Introduction}
In optical communications and other digital transmission systems the ratio between the error probabilities of type $1\to 0$ and $0 \to 1$ can be large~\cite{blaum1994codes}. 
Practically, one can assume that only one type of error can occur. These channels are called asymmetric.
  In this paper, we discuss the problem of finding coding strategies for the Z-channel with feedback. The Z-channel 
is of asymmetric nature because it permits an error $1\to 0$, whereas it prohibits an error $0\to 1$ (see Figure~\ref{fig:binary_Z_channel}). Transmission is referred to as being error-free if the output symbol matches the input symbol of the respective symbol transmission.

The Z-channel error model in this paper is purely combinatorial and can be seen as an extension of the model for the zero-error capacity which was introduced in \cite{shannon1956zero}. In this setting, we limit the fraction of erroneous symbols by $\tau = t/n$, where $n$ denotes the blocklength and $t$ is the maximum number of errors within a block. This is in contrast to the probabilistic setting, in which the error probability of the channel is fixed. 
Feedback codes achieving the Z-channel capacity are considered in \cite{tallini2008feedback}.
The figure of merit examined in this work is the maximum asymptotic rate, written as $R(\tau)$ and also called capacity error function \cite{ahlswede2006non}, which we define to be the maximum rate at which information can be communicated over a channel error-free as the blocklength $n$ goes to infinity in the aforementioned combinatorial setting.

The problem of finding encoding strategies for the Z-channel using noiseless feedback is equivalent to a variation of Ulam's game, the half-lie game.
The first appearance of the half-lie game occurs in \cite{rivest1980coping}. In this game for two players one player, referred to as Paul, tries to find an element $x \in \mathcal{M}$ by asking $n$ yes-no questions which are of the form: Is $x\in A$ for some $A \subseteq \mathcal{M}$? The other player, the responder Carole, knows $x$ and is allowed to lie at most $t$ times if the correct answer to the question is yes. In comparison to the original Ulam game \cite{ulam1991adventures}, Carole is not allowed to lie if the correct answer is no. Before Ulam proposed the game it was already described by Berlekamp
\cite{berlekamp1968block} and by Renyi \cite{renyi61}. For a survey
of results see \cite{Cicalese13}.
It is known that for fixed $t$, the cardinality of the maximal set $\mathcal{M}$ is asymptotically  $2^{n+t}t! n^{-t}$
for Paul to win the half-lie game. First this was shown for $t=1$ in~\cite{cicalese2000optimal} and later generalized in~\cite{dumitriu2004halfliar,spencer2003halflie} for arbitrary $t$.
Due to the equivalence of the half-lie game and the coding problem with feedback for the Z-channel, the coding problem has been solved for an arbitrary but fixed number of errors for the asymptotic case when $n$ goes to infinity.


Finally, let us refer to the most relevant papers to our work. By random arguments, it was proved in~\cite{d1975upper} that $R(\tau)>0$ for $\tau<1/2$. In \cite{deppe2020bounds} a feedback strategy based on the rubber method \cite{ahlswede2006non} was introduced to find an encoding strategy achieving a positive asymptotic rate for any $\tau<1/2$. The corresponding lower bound on $R(\tau)$  is plotted in green  on Figure~\ref{fig:asymptotic rate}.
\subsection{Our contribution}
In this paper, we develop new efficient encoding algorithms for the Z-channel with feedback. In particular, we provide a family of error-correcting codes with the asymptotic rate $(1+\tau) - (1+\tau)\log(1+\tau) + \tau \log \tau$, which is positive for any $\tau<1$ and improves the result from~\cite{deppe2020bounds} in all but countable number of points. The corresponding lower bound on $R(\tau)$ is shown in blue in Figure~\ref{fig:asymptotic rate}. Additionally, we prove an upper bound on $R(\tau)$, which is depicted as the dashed line.
\subsection{Outline}
The remainder of the paper is organized as follows. In Section~\ref{ss::formal definitions}, we formally define the problem of coding with feedback over the Z-channel and introduce some auxiliary terminology. In
Section~\ref{ss::optimal rate}, we provide our encoding algorithm achieving a positive asymptotic rate for any fraction of errors $\tau<1$, which gives rise to our main result, Theorem~\ref{th::optimal rate}. An upper bound on the asymptotic rate is proposed in Section~\ref{ss::upper bounds}. Finally, Section~\ref{ss::conclusion} concludes the paper. 
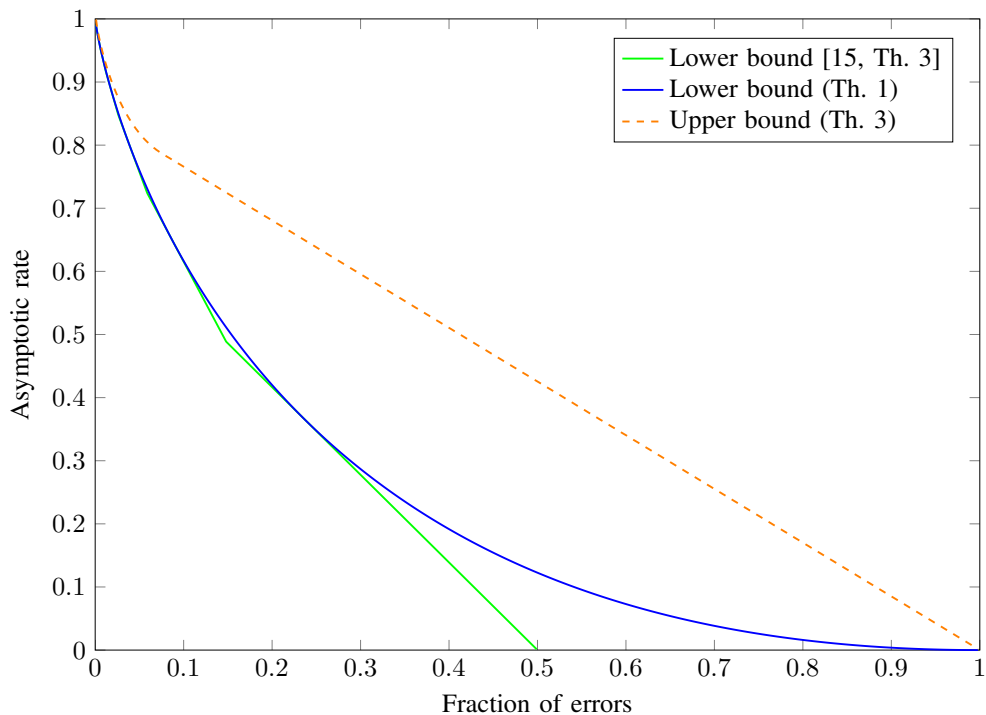
\begin{figure}
\centering
\begin{tikzpicture}[thick,scale=0.9]
\pgfplotsset{compat = 1.3}
\begin{axis}[
	legend style={nodes={scale=1, transform shape}},
	legend cell align={left},
	width = 0.8\columnwidth,
	height = 0.6\columnwidth,
	xlabel = {Fraction of errors},
	xlabel style = {nodes={scale=1, transform shape}},
	ylabel = {Asymptotic rate},
	ylabel style={nodes={scale=1, transform shape}},
	xmin = 0,
	xmax = 1,
	ymin = 0,
	ymax = 1,
	legend pos = north east]

\addplot[color=green, mark=none,thick] table {rub.txt};
\addlegendentry{Lower bound~\cite[Th.~3]{deppe2020bounds}}

\addplot[color=blue, mark=none,thick] table {low.txt};
\addlegendentry{Lower bound (Th.~\ref{th::optimal rate})}


\addplot[color=orange,dashed, mark=none,thick] table {up.txt};
\addlegendentry{Upper bound (Th.~\ref{th::non-trivial upper bound})}

\end{axis}
\end{tikzpicture}
\caption{Maximum asymptotic rate of error-correcting codes for the Z-channel with noiseless feedback.}
  \label{fig:asymptotic rate}
\end{figure}

\section{Preliminaries}\label{ss::formal definitions}
A transmission scheme with feedback enables the sender to choose his encoding strategy in a way that makes use of the knowledge about previously received symbols at the decoder. This is shown in Figure~\ref{fig:channel_feedback}.
Let $\mathcal{M}$ denote the set of messages. The sender chooses one of them, say $m$, which he wants to send to the receiver.
An encoding algorithm for a feedback channel of blocklength $n$ is composed of a set of functions
$$
c_i : \mathcal{M} \times \{0,1\}^{i-1} \rightarrow \{0,1\}, \quad i \in \{1,\dots,n\}.
$$
The encoding algorithm is then constructed as
\begin{equation}\label{def::encoding algorithm}
    c(m,y^{n-1}) = ((c_1(m), c_2(m,y_1),\dots,c_n(m,y^{n-1})),
\end{equation}
where $y^k\defeq(y_1,\dots,y_k)$ with $y_i$ being the $i^{th}$ received symbol. Moreover, the set of possible values for the received symbol $y_i$ conditioned on $c_i$ is defined by the channel, in our case the Z-channel depicted in Figure~\ref{fig:binary_Z_channel}.
Suppose that at most $t$ errors occur within a block of length $n$. For $m\in\mathcal{M}$, we define the set of output sequences for an encoding strategy by
\begin{equation*}
    \mathcal{Y}_t^n(m) \defeq \left\{y^n\in\{0,1\}^n : \,y_i\le c_i(m,y^{i-1}),\ d_H(y^n, c(m,y^{n-1}))\leq t \right\} \enspace,
\end{equation*}
where $d_H(\cdot,\cdot)$ denotes the Hamming metric. Additionaly, we denote the Hamming weight of a sequence $x$ by $w_H(x)$.

\begin{figure}[t]
  \centering
  \begin{minipage}[t]{.3\textwidth}
  \centering
	\begin{tikzpicture}
	\node (A) at (0,0) {0};
	\node (C) at (3,0) {0};
	\node (B) at (0,-3) {1};
	\node (D) at (3,-3) {1};
	\path[line] (A) -- (C);
	\path[line] (B) -- (D);
	\path[line] (B) -- (C);
	\end{tikzpicture}
	\captionof{figure}{Z-channel}
	\label{fig:binary_Z_channel}
\end{minipage}%
\begin{minipage}[t]{.7\textwidth}
  \centering
  \begin{tikzpicture}
    \node[block_large] (c) at (0,0) {Channel};
    \node[block] (d) at (3.5,0) {Decoder};
    \node[block] (e1) at (-3,0) {Encoder};
    \node (b) at (-5,0) {};
    \node (u1) at (5.5,0) {};
    \path[line] (b) -- node[near start, above] {$m$} (e1);
    \path[line] (e1.east) -- node[above] {$c_i$} (c.west);
    \path[line] (c) -- node[near start, above] {$y_i$} (d);
    \path[line] (2,0) -- (2,2) -- (-3.5,2) -- ([xshift=-0.5cm]e1.north);
    \path[draw, dashed, rounded corners] (-4.25,2.5) -- (-4.25,-1.5) -- (4.75,-1.5) -- (4.75,2.5) -- cycle;
    \path[line] (d.east) -- node[near end, above] {$\hat{m}$} (u1);
    \node[draw, circle, minimum size=1mm, inner sep=0pt, outer sep=0pt, fill=darkgray] at (2,0) {};
\end{tikzpicture}
  \captionof{figure}{Channel with feedback}
  \label{fig:channel_feedback}
\end{minipage}
\end{figure}

\begin{definition}
    An encoding strategy~\eqref{def::encoding algorithm} is called \textit{successful} if $\mathcal{Y}_t^n(m_1) \cap \mathcal{Y}_t^n(m_2) = \emptyset$ for all $m_1,m_2\in\mathcal{M}$ with $m_1\neq m_2$.
\end{definition}
\begin{definition}
	Let $M(n,t)$ denote the maximum number of messages in $\mathcal{M}$ for which there exists a successful encoding strategy.  A successful encoding strategy for $M(n,t)$ messages is said to be \textit{optimal}.
\end{definition}
Using this terminology, let us recall the well-known result on $M(n,t)$ for a fixed integer $t$.
\begin{lemma}[Follows from~\cite{spencer2003halflie}]
     Given a fixed integer $t\ge 1$, the maximum number of messages is asymptotically
     $$
     M(n,t)= \frac{2^{n+t}t!}{n^t}(1+o_t(1)) \quad \text{as }n\to\infty.
     $$
\end{lemma}
Unlike studies in~\cite{spencer2003halflie,dumitriu2004halfliar,cicalese2000optimal}, we discuss the case when $t$ is linear with $n$. In this setting, it is natural to investigate the exponential growth of the quantity $M(n,t)$. Hereafter, we write $\log $ to denote the logarithm in base two.
\begin{definition}\label{def::asymptotic rate}
	For any $\tau$ with $0\le \tau \le 1$, we define the \textit{maximum asymptotic rate} of optimal feedback encoding strategies capable of correcting a fraction $\tau$ of asymmetric errors to be
	$$
	R(\tau) \defeq \limsup_{n\to\infty} \frac{\log(M(n,\lceil \tau n\rceil))}{n}.
	$$
\end{definition}
\section{Lower Bound on $R(\tau)$}\label{ss::optimal rate}
In this section we give a successful encoding strategy for the Z-channel. This gives a lower bound  on the maximum asymptotic rate $R(\tau)$ of optimal feedback encoding strategies capable of correcting a fraction $\tau$ of asymmetric errors.

\begin{theorem}\label{th::optimal rate}
	For any $\tau$, $0\le\tau\le 1$, we have 
	$$
	R(\tau) \ge \underline{R}(\tau)\defeq (1+\tau)-(1+\tau)\log(1+\tau) + \tau \log\tau.
	$$
\end{theorem}
\subsection{Encoding strategy}\label{sec::encoding startegy}
At the start of the message transmission the receiver only knows the set of messages $\mathcal{M}$. The sender chooses a message $m\in \mathcal{M}$. The goal of an encoding strategy is to reduce the number of \textit{eligible} messages from the receiver's viewpoint until only one message $m$ is left. The encoding algorithm we provide divides the number of channel uses $n$ into subblocks. Therefore, the encoding procedure is potentially divided into several steps. We denote the set of eligible messages from the receiver's perspective after the $i^{th}$ step as $\mathcal{M}_{i+1}$ with $M_{i+1}=|\mathcal{M}_{i+1}|$,  the number of remaining channel uses as $n_{i+1}$ and the maximum number of possible errors  as $t_{i+1}$. 


In the following the algorithm depicted in Figure~\ref{fig::algorithm} is described. Before every step, the sender (as well as the receiver) checks the following two properties
\begin{equation}\label{eq::prop1}
t_i = 0
\end{equation}
and
\begin{equation}\label{eq::prop2}
\lvert\mathcal{M}_i\rvert \le n_i-t_i+1.
\end{equation}
    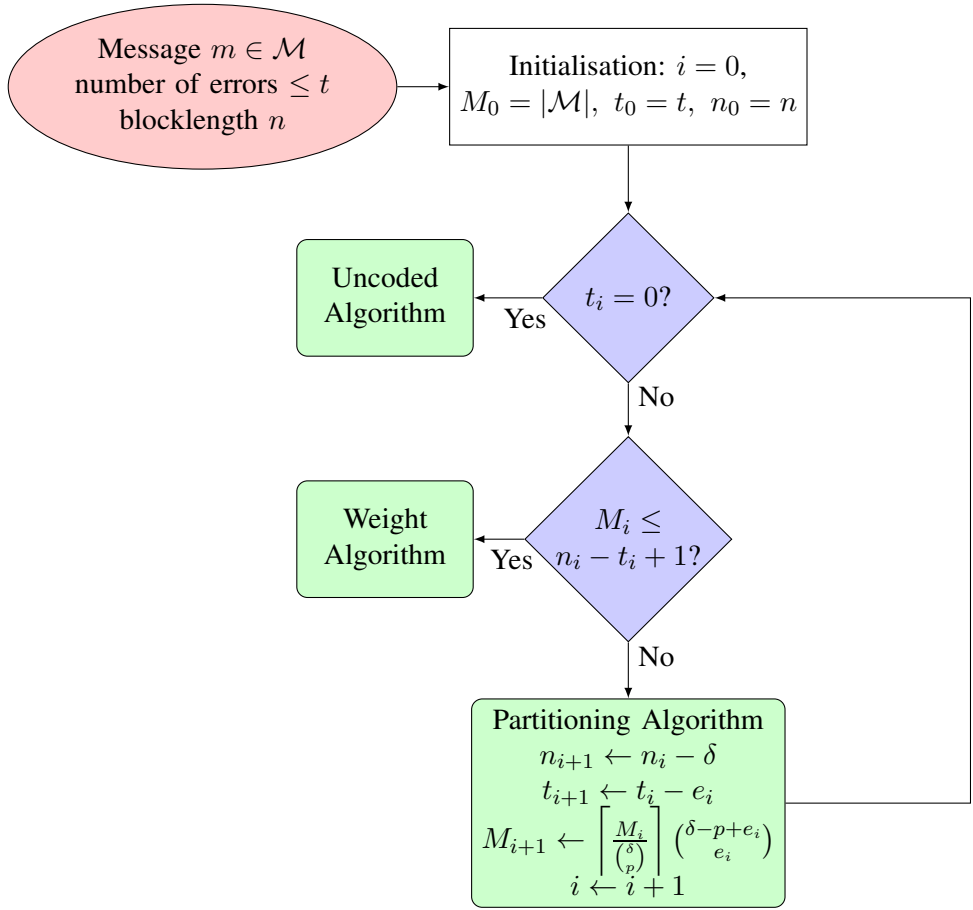
\begin{figure}
    \centering
\begin{tikzpicture}[node distance = 3cm, auto]
    \node [initialization,align=center] (init) {Initialisation: $i=0$, \\ $M_0=|\mathcal{M}|,\ t_0=t, \ n_0 = n$};
    \node [cloud, left of=init, align=center, node distance=5.6cm] (start) {Message $m\in\mathcal{M}$  \\ number of errors $\le t$\\ blocklength $n$};
    \node [decision, below of=init, node distance=2.8cm] (decide1) {$t_i=0$?};
    \node [decision, below of=decide1, node distance=3.2cm] (decide2) {$M_i\le n_i- t_i + 1$?};
    \node [algorithm, left of=decide1, node distance=3.2cm, align=center] (strategy3) {Uncoded \\ Algorithm};
    \node [algorithm, left of=decide2, node distance=3.2cm, align=center] (strategy2) {Weight \\ Algorithm};
    \node [algorithm, below of =decide2, node distance=3.5cm, align=center] (strategy1) {Partitioning Algorithm \\ $n_{i+1}\gets n_i-\delta$ \\ $t_{i+1}\gets t_{i} - e_i$ \\ $M_{i+1}\gets \left\lceil\frac{M_i}{\binom{\delta}{p}}\right\rceil\binom{\delta - p + e_i}{e_i}$ \\
    $i\gets i+1$};
    \path [line] (start) -- (init);
    \path [line] (init) -- (decide1);
    \path [line] (decide1) -- node [near start] {No} (decide2);
    \path [line] (decide1) -- node [near start] {Yes} (strategy3);
    \path [line] (decide2) -- node [near start] {No} (strategy1);
    \path [line] (decide2) -- node [near start] {Yes} (strategy2);
    \path[line] (strategy1.east) -- (4.5,-9.5) -- (4.5,-2.8) -- (decide1.east);
\end{tikzpicture}\caption{Encoding algorithm for transmission over the Z-channel}\label{fig::algorithm}
\end{figure}
Depending on which of them hold the sender chooses one out of three algorithms for encoding.
If both conditions~\eqref{eq::prop1}-\eqref{eq::prop2} do not hold, then the sender makes use of Partitioning Algorithm. This strategy tries to limit the set of eligible messages by dividing the message space into subsets and sending the index of the subset containing the message. After this subblock transmission the sender and the receiver examine the conditions~\eqref{eq::prop1}-\eqref{eq::prop2} and check whether the encoding and decoding strategies have to be adjusted for the remainder of the block.

If property~\eqref{eq::prop1} is violated and property~\eqref{eq::prop2} holds, then the sender uses the Weight Algorithm for information transmission in the remaining channel uses.
 
If the condition~\eqref{eq::prop1} is true, then the sender applies the Uncoded Algorithm for information transmission in the remaining channel uses.
Below we describe the three algorithms required for our encoding strategy. \medskip

\textbf{Partitioning Algorithm:}
    This algorithm relies on the specific choice of positive fixed integers $\delta$ and $p$ with $\delta>p$. We partition the message space before the $i^{th}$ step $\mathcal{M}_i$ into $\binom{\delta}{p}$ subsets $\mathcal{M}_{i,j}$ of almost equal sizes
    $$
        \mathcal{M}_i = \bigcup_{j=1}^{\binom{\delta}{p}} \mathcal{M}_{i,j} \enspace .
    $$
    The size of each subset is either $\left\lceil M_i / \binom{\delta}{p}\right\rceil$, or $\left\lfloor M_i / \binom{\delta}{p}\right\rfloor$.
    The exact way in which the message space is to be partitioned is to be agreed between the sender and the receiver before the data transmission. Then the sender finds the index of the group containing the message and transmit this index using a subblock of length $\delta$ containing $p$ ones. In this way the receiver can determine the number of errors inflicted by the channel within this subblock by counting the number of ones. There are $p+1$ possible cases depending on the number of errors $e_i$ within the respective subblock of the $i^{th}$ step.
    
    If $e_i=p$ errors occur, the message space consistent with the outcome of the channel is not changed and the receiver obtains the information that $\mathcal{M}_{i+1}= \mathcal{M}_i$, $n_{i+1} = n_i-\delta$ and $t_{i+1} = t_i - p$.
    
	When $e_i<p$ errors occur, there are ${\delta -p + e_i \choose e_i}$ subsets of messages $\mathcal{M}_{i,k}$ that are consistent with the outcome of the Z-channel. $\M_{i+1}$ is then equal to the union of these subsets. Therefore, the set of eligible messages in accordance with the received $\delta$ symbols is reduced and we have $M_{i+1}\le \left\lceil M_i / \binom{\delta}{p}\right\rceil \binom{\delta -p + e_i}{e_i}$. Moreover, the receiver and the sender obtain $n_{i+1} = n_i - \delta$ and $t_{i+1} = t_i - e_i$.
	
    \medskip
    
    \textbf{Weight Algorithm:} We order the messages within the set of eligible messages, denoted by $\mathcal{M}'$, by enumerating them: $\mathcal{M'} = \{m_0, m_1, \dots, m_{\lvert \mathcal{M'} \rvert-1}\}$. The sender would like to transmit one of the messages of $\mathcal{M}'$, say $m_k$, to the receiver by using the channel $n'$ times. To this end, the sender transmits a $1$ over the channel until a $1$ is received exactly $k$ times. This happens at some point if a sufficient amount of channel uses is considered because the number of errors is limited. We denote this limit as $t'$. After that, the sender transmits $0$-symbols which cannot be disturbed by the Z-channel. 
    The receiver finds the Hamming weight $w$ of the received sequence of length $n'$ and outputs the message $m_w$. This strategy is successful, i.e., $m_k=m_w$, if $|\mathcal{M}'|\leq n'-t'+1$.  \medskip
    
    \textbf{Uncoded Algorithm:} We denote the ordered set of eligible messages as $\mathcal{M}' = \{m_0, m_1, \dots, m_{\lvert \mathcal{M'} \rvert-1}\}$. The senders task is to send one of the messages, say $m_k$ to the receiver by using the channel $n'$ times.  In order to do so, it sends the (standard) binary representation of the index $k$ over the channel. This strategy is successful if the sender is allowed to use the channel at least $\lceil\log \mathcal{M}'\rceil$ times and no errors occur.  \medskip
 
We shall prove that for any proper choice of integers $\delta$, $p$, $k$ and $t$ and real $\epsilon>0$, the sender can have the message set $\mathcal{M}$ of size at least
$$
\left\lceil \frac{{\delta \choose p}}{\epsilon} \right\rceil  \frac{(1-\epsilon)^k{\delta \choose p}^k}{{\delta k -pk + t + p \choose t+p}} - 1
$$
for $n=\left\lceil{\delta \choose p}/\epsilon \right\rceil + \delta k$ channel uses, having at most $t$ errors. More formally, it is shown in Lemma~\ref{lem::strong statement}.
\subsection{Analysis of the proposed algorithm}
\begin{lemma}\label{lem::strong statement}
 	Let $k$, $\delta$, and $p$ be positive integers such that $p < \delta$ and let $t \geq 0$. Let $\epsilon>0$ be a fixed real number such that 
 $$
 	    \gamma_e \defeq (1-\epsilon) \frac{\binom{\delta}{p}}{\binom{\delta-p+e}{e}} > 1, \quad \forall e \in \{0,\ldots,p-1\}.
 $$
 We define $\gamma_p\defeq 1$, $A\defeq \left\lceil \binom{\delta}{p}/\epsilon \right\rceil$
 and the set
 	$$
 	    \mathcal{S}_{(k,t,p)} \defeq \left\{(k_0,k_1,\ldots,k_p) \in \mathbb{N}_0^{p+1}:\quad \sum_{i=0}^p k_i =k,\, \sum_{i=0}^p ik_i \leq t+p \right\}.
 	$$
 	Then for any non-negative integers $t$ and $k$ such that $\delta k\ge t$, we have 
	\begin{equation}\label{eq::finite length}
	M(A + \delta k, t)\ge  \left\lfloor A\min_{\mathcal{S}_{(k,t,p)}} \prod_{e=0}^{p} \gamma_e^{k_e} \right\rfloor.
\end{equation}
In particular, it follows that
\begin{equation}\label{in::construction}
M(A+\delta k,t)\ge  A  \frac{(1-\epsilon)^k{\delta \choose p}^k}{{\delta k -pk + t + p \choose t+p}} - 1.
\end{equation}
\end{lemma}
\begin{remark}
	An integer $k$ in Lemma~\ref{lem::strong statement} corresponds to the maximum number of times when Partitioning Algorithm is performed. We can think about the set $\mathcal{S}_{(k,t,p)}$ as the set of all possible error distributions that may happen when the sender transmits the message. Specifically, a tuple $(k_0,k_1,\ldots,k_p)$ means that during $n$ transmissions there will be $k_i$ blocks of $\delta$ channel uses for which exactly $i$ errors occur. We also note that the set $\mathcal{S}_{(k,t,p)}$ includes the tuple $(k,0,\ldots,0)$ and, thus, the minimization in~\eqref{eq::finite length} is well defined. 
\end{remark}
\begin{proof}

	We shall prove this lemma by applying the encoding algorithm described above and using induction on the sum $k+t$ in the pair $(k,t)$. The base cases when $k=\lceil t/\delta\rceil$ or $t=0$ follow from the description of Weight Algorithm and Uncoded Algorithm. Indeed, for $t=0$, we get the minimum in~\eqref{eq::finite length} is at most $A\gamma_0^k\le A{\delta \choose p}^k\le A 2^{\delta k}$. However, according to the encoding algorithm the sender has to use Uncoded Algorithm and can have the message set of size $2^{A+\delta k}$ which is larger than $A 2^{\delta k}$. When $k=\lceil t/\delta\rceil$, then the minimum in~\eqref{eq::finite length} is equal to $A$ and is attained with $(k_0,k_1,\ldots,k_p)=(0,\ldots, 0,k)$ as 
	$$
	\sum_{i=0}^p k_i = k,\quad \sum_{i=0}^p ik_i = p\lceil t/\delta\rceil \le pt/\delta + p\le t+p.
	$$
	Due to the above algorithm, the sender makes use of Weight Algorithm at this point. According to its description, we can transmit $A + \delta k -t +1$ which is larger than $A$ for $k=\lceil t/\delta\rceil$.
	
	This concludes the base cases of the induction. In the following we show the inductive step.
	
	We prove that the sender can transmit  
	$$
	M \defeq  \left\lfloor A\min_{\mathcal{S}_{(k,t,p)}} \prod_{i=0}^{p} \gamma_i^{k_i} \right\rfloor.
	$$
	messages using $A+k\delta$ channel uses when at most $t$ errors may occur. We note that the set $\mathcal{S}_{(k,t,p)}$ includes $(k,0,\ldots,0)$. Thus, the minimization in~\eqref{eq::finite length} is well defined. According to the above algorithm, the sender check two conditions~\eqref{eq::prop1}-\eqref{eq::prop2}. If~\eqref{eq::prop1} is true, then we are in a base case. If~\eqref{eq::prop2} holds, then the sender makes use Weight Algorithm and successfully transmit the message.	If both conditions are failed, then the sender uses Partitioning Algorithm. It remains to check that for any $e\in\{0,1,\ldots,p\}$ (the number of errors in the block of $\delta$ bits), the sender would be able to transmit the message out of
	$$
	\begin{cases}
	\left\lceil \frac{M}{{\delta \choose p}}\right\rceil {\delta -p + e\choose e},\quad&\text{for }e\in\{0,1\ldots,p-1\}\\
	M,\quad &\text{for }e=p
	\end{cases}
	$$
	ones using remaining $A+\delta (k-1)$ channel uses and having at most $t-e$ errors. By the inductive hypothesis we are able to transmit
$$
  \left\lfloor A\min_{\mathcal{S}_{(k-1,t-e,p)}}  \prod_{i=0}^{p} \gamma_i^{k_i}\right\rfloor	
$$
messages for this set-up. Thus it remains to check the inequality
	\begin{equation}\label{eq::final inequality}
\left\lceil \frac{M}{{\delta \choose p}}\right\rceil {\delta -p + e \choose e} \le \left\lfloor A\min_{\mathcal{S}_{(k-1,t-e,p)}}  \prod_{i=0}^{p} \gamma_i^{k_i}\right\rfloor,\quad \forall e \in \{0,1\ldots,p-1\} \enspace . 
\end{equation}
	Let us elaborate on the left-hand side of the inequality
	\begin{align}
	\left\lceil \frac{M}{{\delta \choose p}}\right\rceil {\delta -p + e \choose e} &\le \frac{M p! (\delta - p+e)!}{\delta! e!} + {\delta -p + e \choose e} \nonumber\\
	&\overset{(a)}{\le} \frac{M p! (\delta - p+e)!}{\delta! e!} + A \epsilon \overset{(b)}{=} \frac{M (1-\epsilon)}{ \gamma_e} + A \epsilon, \label{in::simple bound}
	\end{align}
	where we used the property $A\epsilon = \lceil{\delta \choose p} / \epsilon \rceil \epsilon \ge {\delta - p +e \choose e}$ for any $e\in\{0,1,\ldots, p\}$ in $(a)$ and $\gamma_e = (1-\epsilon)\frac{\delta! e!}{p! (\delta - p+e)!}$ in $(b)$.
	We note that
	\begin{equation}\label{eq::evident bound on M}
M \le \left\lfloor A\gamma_e\min_{\mathcal{S}_{(k-1,t-e,p)}}  \prod_{i=0}^{p} \gamma_i^{k_i}\right\rfloor.
	\end{equation}
	This inequality holds because
	$$
	    \gamma_e \min_{\mathcal{S}_{(k-1,t-e,p)}}  \prod_{i=0}^{p} \gamma_i^{k_i} \geq \min_{\mathcal{S}_{(k,t,p)}}  \prod_{i=0}^{p} \gamma_i^{k_i}
	$$
	as we get the left-hand side of the inequality from the right-hand side by adding the additional constraint  that $k_e \geq 1$ to the minimization.
	Thus, combining the inequalities~\eqref{in::simple bound}-\eqref{eq::evident bound on M}, we get 
		\begin{align*}
	\left\lceil \frac{M}{{\delta \choose p}}\right\rceil {\delta - p+e \choose e} &\le  A(1-\epsilon)\min_{\mathcal{S}_{(k-1,t-e,p)}}  \prod_{i=0}^{p} \gamma_i^{k_i} + A \epsilon\\
			&\overset{(c)}{\le} A\min_{\mathcal{S}_{(k-1,t-e,p)}}  \prod_{i=0}^{p} \gamma_i^{k_i}.
	\end{align*}
	To prove~$(c)$, we observe that $\gamma_i\ge 1$ for all $i$.
	As we compare an integer in the left-hand side with a real number in right-hand side, we can apply the floor operation to the latter. This proves~\eqref{eq::final inequality} and completes the proof of the inductive step.
	
	It remains to show~\eqref{in::construction}. Suppose that the minimum in~\eqref{eq::finite length} is attained with $k_0=k_0',k_1=k_1',\ldots,k_p=k_p'$ such that 
	$$
	\sum_{i=0}^{p} k_i' = k,\quad \sum_{i=0}^{p}i k_i' \eqdef t'\le t+p.
	$$
	Then we derive
\begin{align*}
	M &= \left\lfloor A\prod_{i=0}^{p} \gamma_i^{k_i'} \right\rfloor \ge  A(1-\epsilon)^k\prod_{i=0}^{p} \frac{{\delta \choose p}^{k_i'}}{{\delta - p +i \choose i}^{k_i'}} -1
	\\
	&=  A \frac{(1-\epsilon)^k{\delta \choose p}^{k}}{\prod_{i=0}^{p}{\delta - p+i \choose i}^{k_i'}}  -1
	\\
	&\overset{(d)}{\ge}  A \frac{(1-\epsilon)^k{\delta \choose p}^{k}}{{\delta k - pk +t' \choose t'}}  -1 
	\\
	&\overset{(e)}{\ge}  A \frac{(1-\epsilon)^k{\delta \choose p}^{k}}{{\delta k -pk +t+p \choose t+p}} -1,
\end{align*}
where the property ${u \choose v} {w \choose z} \le {u+w \choose v+z}$ yields~$(d)$ and the monotonicity of the function ${u +x \choose x}$ in $x$ implies~$(e)$.
\end{proof}
Finally we are in a good position to prove Theorem~\ref{th::optimal rate}.
\begin{proof}[Proof of Theorem~\ref{th::optimal rate}] 
Let us fix some positive $\tau<1$.
For any $\epsilon_R>0$ and small enough $\epsilon_\tau > 0$, we shall prove the existence of a code of an arbitrary large blocklength and code rate at least $\underline{R}(\tau)-\epsilon_R$ capable of correcting a fraction $\tau-\epsilon_\tau$ of errors.

In what follows, we vary positive integers $k$ and $\delta$ with $k>\delta$.  Define $t=\lceil \tau k \delta\rceil$, $p=\lfloor \delta(1/2 + \tau /2) \rfloor $, $A=\lceil {\delta \choose p}/\epsilon \rceil$  and $n= A+\delta k$, where the real parameter $\epsilon$ is fixed and satisfies 
\begin{align}
    0<&\,\epsilon< \frac{1-\tau}{2}\le 1-p/\delta, \nonumber
\\    0<&\,\underline{R}(\tau)+\log(1-\epsilon) - 3\epsilon_\tau. \label{eq::epsilon req}
\end{align} 
Let $\delta_0$ be such that for any $\delta\ge\delta_0(\epsilon_\tau, \tau)$ and $k\ge \delta$, we have
\begin{equation} \label{eq:: delta choose p}
{\delta \choose p} \ge 2^{\delta \left(h\left(\frac{1+\tau}{2}\right)-\epsilon_\tau\right)}
\end{equation}
and
\begin{equation} \label{eq:: delta k - pk + t choose t}
{\delta k -pk +t+p \choose t+p} \le 2^{\delta k\frac{1+\tau}{2} \left(h\left(\frac{2\tau}{1+\tau}\right) + \epsilon_\tau\right)},
\end{equation}
where the binary entropy function $h(x) \defeq - x\log(x) - (1-x)\log(1-x)$. To prove the existence of such $\delta_0$, we note that
$$
\lim_{\delta\to\infty} \frac{p}{\delta} = \lim_{\delta\to\infty} \frac{\lfloor \delta(1/2 + \tau /2) \rfloor}{\delta} =\frac{1+\tau}{2},
$$
and for $k\ge \delta$,
$$
\quad \lim_{\delta\to\infty} \frac{t+p}{\delta k - pk + t +p}=\frac{2\tau}{1+\tau},
$$
and for any integers $u> v \ge 1$, the binomial coefficient $\binom{u}{v}$ satisfies
\begin{equation}\label{eq::binomial coefficient}
\sqrt{\frac{u}{8v(u-v)}}2^{uh(v/u)} \leq \binom{u}{v} \leq \sqrt{\frac{u}{2\pi v(u-v)}}2^{uh(v/u)}.
\end{equation}
Then we take $k_0=k_0(\delta, \tau,\epsilon_\tau)$ such that for any $k\ge k_0$, the fraction of errors
\begin{equation*}\label{eq::fraction of errors}
 \frac{t}{n}=\frac{t}{A+\delta k} \ge \tau - \epsilon_\tau   
\end{equation*}
and the blocklength
\begin{equation}\label{eq:n}
n = A+\delta k \le  \delta k(1+\epsilon_\tau)
\end{equation}
and
$$
\frac{(1-\epsilon)^k{\delta \choose p}^{k}}{{\delta k -pk +t+p \choose t+p}} \ge 2.
$$
The latter can be achieved because of the choice of $\epsilon$ in~\eqref{eq::epsilon req} and large enough $\delta$ in~\eqref{eq:: delta choose p}-\eqref{eq:: delta k - pk + t choose t}.
By Lemma~\ref{lem::strong statement}, there exists a feedback error-correcting code with blocklength $n=A+\delta k$ for a message space of size
\begin{equation}\label{eq::size of the code}
M\ge A \frac{(1-\epsilon)^k{\delta \choose p}^{k}}{{\delta k -pk +t+p \choose t+p}} -1 \ge \frac{(1-\epsilon)^k{\delta \choose p}^{k}}{2{\delta k -pk +t+p \choose t+p}},
\end{equation}
capable of correcting $t$ errors when transmitted through the Z-channel. Thus, combining~\eqref{eq:: delta choose p}-\eqref{eq::size of the code} yields
\begin{align*}
R(\tau -\epsilon_\tau)&\ge \frac{\log M}{A+\delta k} \\
&\ge \frac{k\log(1-\epsilon)+\delta k \left(h\left(\frac{1+\tau}{2}\right)-\epsilon_\tau - \frac{1+\tau}{2}h\left(\frac{2\tau}{1+\tau}\right) -\epsilon_\tau\right)-1}{\delta k(1+\epsilon_\tau)}\\
&=(1+\tau)\log\left(\frac{2}{1+\tau}\right) + \tau \log \tau -\epsilon_R,
\end{align*}
where
$$
\epsilon_R\le -\log(1-\epsilon) + 3\epsilon_\tau + \frac{1}{\delta k}.
$$
As  $\epsilon$ and $\epsilon_\tau$ can be taken as small as needed and $\delta$ and $k$ can be arbitrary large, the statement of Theorem~\ref{th::optimal rate} follows.
\end{proof}

\subsection{Complexity of encoding and decoding strategy}
In this section we will elaborate further on the choice of parameters used within the encoding strategy presented in Section~\ref{sec::encoding startegy}. We will show that it is possible to choose the parameters in such a way that the complexity of the encoding algorithm becomes $\mathcal{O}(n)$. 
\begin{theorem}
    For any $\tau$, $0<\tau<1$, there exist successful encoders and decoders for the Z-channel with feedback having complexity $\mathcal{O}(n)$ and achieving the rate $\underline{R}(\tau)$ specified in Theorem \ref{th::optimal rate} arbitrarily close.
\end{theorem}
\begin{proof}
It is clear that the encoding and decoding complexity of both Weight Algorithm and Uncoded Algorithm is $\mathcal{O}(n)$. Thus, it remains to concentrate on the analysis of  Partitioning Algorithm.  

Remember that if $t_i>0$ and $M_i > n_i-t_i+1$ the encoder uses Partitioning Algorithm to reduce the set of eligible messages. Recall that $M_i$ denotes $|\mathcal{M}_i|$. To prove this statement we consider the message $m$ to be an index within the set $\{1,\dots,M_i\}$. Within the partitioning algorithm it was only specified that the set of eligible messages is to be split into subsets $\mathcal{M}_{i,j}$, of cardinality $\left \lceil M_i /\binom{\delta}{p} \right\rceil$ or $\left\lfloor M_i /\binom{\delta}{p} \right\rfloor$. 
We consider the set of messages $\mathcal{M}_i$ to be points on a line which are in ascending order according to the message index.
We graphically illustrate the set of eligible messages as a set of discrete points on a straight line which we partition into $\binom{\delta}{p}$ segments $\mathcal{M}_{i,j}$ in the natural way depicted on Figure~\ref{fig:partitioning_algorithm}.
\begin{figure}[h!]
    \centering
    \begin{tikzpicture}
 
    \def\xdist{7}
    \def\ydist{1}
    \def\leng{0.5}
    
    \coordinate (zero) at (0,0);
 
      
    \draw[thick,color=black] (0,0) -- (30*\leng,0);
    \foreach \x in {0,...,10}
        \draw (\x*3*\leng,0.5*\leng)--(\x*3*\leng,-0.5*\leng);
    
    \draw[fill=black] (17*\leng,0) circle[radius=3pt];
    \draw[thick,->] (16*\leng,1.5*\leng) -- (16.7*\leng,0.5*\leng);
    \node at (15.9*\leng,1.8*\leng) {$m$};
    \foreach \x in {3,6,8}
        \draw [thick, color=green] (\x*\leng*3-3*\leng,-0.75*\leng)--(\x*\leng*3,-0.75*\leng);
    \foreach \x in {1,2,4,5,7,9,10}
        \draw [thick, color=red] (\x*\leng*3-3*\leng,-0.75*\leng)--(\x*\leng*3,-0.75*\leng);
    
    \draw[thick,color=green] (10.5*\leng,-3*\leng) -- (19.5*\leng,-3*\leng);
    
    \draw[->] (7.5*\leng,-1*\leng) -- (11*\leng,-2.5*\leng);
    \draw[->] (16.5*\leng,-1*\leng) -- (15*\leng,-2.5*\leng);
    \draw[->] (22.5*\leng,-1*\leng) -- (19*\leng,-2.5*\leng);
    
    \foreach \x in {0,1,2,3}
        \draw (\x*3*\leng+10.5*\leng,-3.5*\leng) -- (\x*3*\leng+10.5*\leng, -2.5*\leng);
    
    \draw [fill=black] (10.5*\leng + 5*\leng,-3*\leng) circle[radius=3pt];
    
    \draw[thick,->] (14.5*\leng, -3.75*\leng)--(14.5*\leng,-6*\leng);
    
    \node at (17*\leng,-5*\leng){repartitioning};
    \draw[thick,color=black] (10.5*\leng,-7*\leng) -- (19.5*\leng,-7*\leng);
    
    \foreach \x in {0,...,10}
        \draw (\x*9/10*\leng+10.5*\leng,-6.5*\leng) -- (\x*9/10*\leng+10.5*\leng,-7.5*\leng);
    
    \draw [fill=black] (10.5*\leng + 5*\leng,-7*\leng) circle[radius=3pt];
    

    
    
    

 \end{tikzpicture}
    \caption{Partitioning algorithm}
    \label{fig:partitioning_algorithm}
\end{figure}
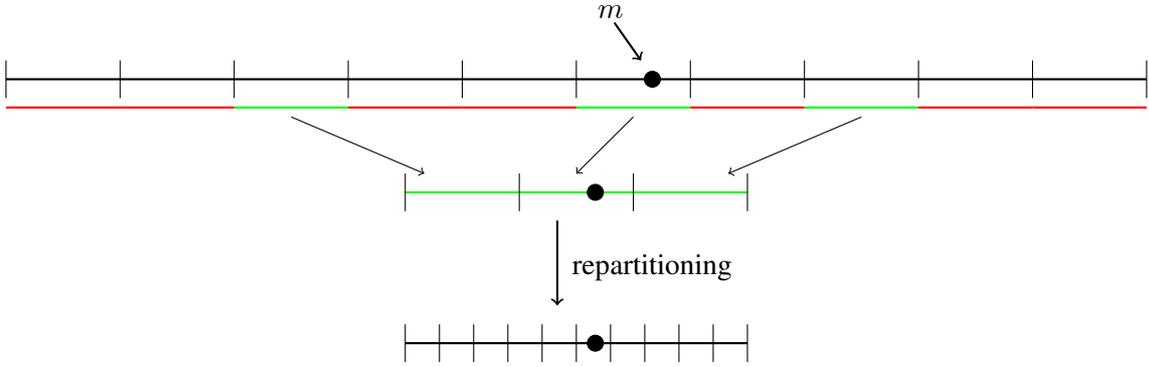

Since $M_i$ may not be divisible by $\binom{\delta}{p}$ we have
\begin{equation}\label{eq:remainder}
    r_i := M_i - \left \lfloor \frac{M_i}{\binom{\delta}{p}} \right \rfloor \binom{\delta}{p}
\end{equation}
subsets of size $\left\lceil \frac{M_i}{\binom{\delta}{p}} \right\rceil$, in the following denoted as large segments, whereas the remaining ones are of size $\left\lfloor \frac{M_i}{\binom{\delta}{p}} \right\rfloor$. We define our partition in such a way that the larger segments are located next to each other starting on the left hand side of the line. Furthermore, due to the fixed weight addressing of subsets within the partitioning algorithm each segment can be identified by a sequence of length $\delta$ containing exactly $p$ ones. The addressing of the segments is in principal arbitrary but has to be known to sender and receiver. We make the simple choice to interpret the addresses as numbers in binary representation and order the addresses in ascending order from left to right.

To perform the partitioning, the encoder transmits the address of the segment containing the message $m$ over the channel. 
If no error occurred the block is of Hamming weight $p$. For every error this Hamming weight is decreased by one. If we denote the number of errors by $e$ there are   $\binom{\delta-p+e}{e}$ segments in accordance with the received partitioning subblock. Those segments are addressed by the sequences which can be created by changing $e$ zeros of the received partitioning subblock to ones. They are referred to in the following as \textit{eligible} segments. The partitioning process is shown in Figure~\ref{fig:partitioning_algorithm} where the eligible segments according to the channel output are underlined in green whereas the segments to be discarded are underlined in red. Due to the feedback the same information can be obtained by the sender, reducing the set of eligible messages.
For the next partitioning step the eligible sets form the new message space. Notice that the order of the messages with respect to each other remains the same even if more than one partitioning step is performed. After each partitioning step the encoder checks whether properties~\eqref{eq::prop1} and~\eqref{eq::prop2} are still both not fulfilled. If this is the case the set of messages is repartitioned and another partitioning step is performed. Otherwise the encoder switches either to the Weight Algorithm or the Uncoded Algorithm, depending on the state of the conditions~\eqref{eq::prop1} and~\eqref{eq::prop2} (see also Figure~\ref{fig::algorithm}).

During potentially numerous partitioning steps it is necessary to keep track of the cardinality of the set of eligible messages after the $i$-th partitioning step $\mathcal{M}_{i+1}$ as well as the relative position of $m$ in within $\mathcal{M}_{i+1}$.

A possible way to do this is the creation of a lookup table with $\binom{\delta}{p}$ entries. The entries within this table contain the $\binom{\delta}{p}$ binary sequences in ascending order when being interpreted as numbers in binary representation. Utilizing the feedback this information is sufficient to compute the cardinality of the set of eligible messages $\mathcal{M}_{i+1}$ after the partitioning step as well as the index of $m$ in $\mathcal{M}_{i+1}$. 
Indeed, by using the lookup table and investigating the received sequence (which is known to the encoder due to the feedback), the sender is able to obtain knowledge about all eligible segments after the $i$-th partitioning step. Furthermore, he knows the lengths of all eligible segments due to the predefined ordering of the segments and his ability to compute $r_i$ using equation~\eqref{eq:remainder}. This is sufficient to compute the cardinality of $\mathcal{M}_{i+1}$. 
Moreover, the encoder is able to obtain the relative index of $m$ within the segment it is contained by using his knowledge of the lengths of the segments $\mathcal{M}_{i,j}$ and the relative index of $m$ in $\mathcal{M}_i$.
Eventually, for the computation of the index of $m$ within $\mathcal{M}_{i+1}$ the encoder uses his knowledge about the ordering of the eligible segments. To obtain the index of $m$ within $\mathcal{M}_{i+1}$ he adds the length of the eligible segments left to the one containing $m$ and the relative index of $m$ within the segment it is contained in.

Next we describe how the decoder obtains knowledge about the message $m$ from the received binary string. The decoder is able to compute the same lookup table for each partitioning step as the encoder does. The strategy of putting the large segments on the left hand side of the line is predefined and known by the decoder. Therefore, the received sequence is sufficient to compute the cardinality of $\mathcal{M}_{i+1}$ after each partitioning step. The number of errors within each received subblock is determined by its Hamming weight.
The computation of $M_{i+1}=|\mathcal{M}_{i+1}|$ as well as the computation of the relative index of $m$ within its containing segment require complexity $\mathcal{O}_\delta(1)$ for each partitioning step at the sender and the receiver side.
The knowledge of $M_{i+1}$ and $t_{i+1}$ enables the decoder to find the point at which the encoder switches his strategy to either the Hamming weight algorithm or the Uncoded algorithm. The index of $m$ within the eligible messages after the final partitioning step can be obtained in a straightforward manner according to the decoding of Weight or Uncoded algorithm. The partitioning steps can then be reversed in accordance to the output sequence to find the message $m$ by plugging the discarded segments back into the line updating the index of $m$ in accordance with the received sequence. After reversing all partitioning steps the decoder knows the index of $m$ within the original message set $\mathcal{M}$.

In the following we discuss in what way $\delta,p$ and $k$ should be chosen such that the complexity of the algorithm is of order $\mathcal{O}(n)$. We need to make sure that the equations \eqref{eq:: delta choose p}, \eqref{eq:: delta k - pk + t choose t} and \eqref{eq:n} can be fulfilled. Those are the main limitations for our choices. According to those limitations we should choose the parameters in a way that the complexity of the encoding strategy is minimized.

So first consider equation \eqref{eq:: delta choose p}.
By looking at equation \eqref{eq::binomial coefficient} we find that equation \eqref{eq:: delta choose p} is fulfilled if
\begin{equation*}
    2^{-\delta \epsilon_\tau} \leq \sqrt{ \frac{\delta}{8\delta (1/2+\tau / 2) \delta (1/2 - \tau / 2)} } = c_1(\tau) \sqrt{\frac{1}{\delta}}
\end{equation*}
for some $c_1(\tau)$ which is a constant depending on $\tau$ but independent of $\delta$. This is equivalent to
\begin{equation}\label{eq:first_condition}
    2^{\delta \epsilon_\tau} \geq \frac{\sqrt{\delta}}{c_1(\tau)}
\end{equation}
and we know that there exists a constant $\delta_0(\tau,\epsilon_\tau)$ independent of $n$ such that inequality~\eqref{eq:first_condition} holds for all $\delta \geq \delta_0(\tau,\epsilon_\tau)$.

Next we consider inequality~\eqref{eq:n}. This inequality is satisfied for $k \in \Theta\left(\frac{n}{\delta}\right)$. It is possible to choose $\delta \geq \delta_0(\tau,\epsilon_\tau)$ to be a constant fulfilling \eqref{eq:first_condition}. If this choice is made, it follows that $k\in \Theta(n)$ and $\delta k \in \Theta(n)$.

Finally we need to consider equation \eqref{eq:: delta k - pk + t choose t}.
Again by equation \eqref{eq::binomial coefficient} we find that equation \eqref{eq:: delta k - pk + t choose t} is fulfilled if
\begin{equation}\label{eq:second_condition}
    \sqrt{\frac{\delta}{2\pi \delta (1/2 + \tau /2) \delta (1/2-\tau / 2)}} = c_2(\tau) \sqrt{\frac{1}{\delta}} \leq 2^{\delta k\frac{1+\tau}{2} \epsilon_\tau},
\end{equation}
where some $c_2(\tau)$ depends only on $\tau$.
Since $k>\delta$ and because for their product it holds $\delta k \in \Theta (n)$ the inequality in \eqref{eq:second_condition} holds.

As $\delta \in \mathcal{O}_{\tau, \epsilon_\tau}(1)$, the creation of the lookup tables for each partitioning step can be done in $\mathcal{O}_{\tau, \epsilon_\tau}(1)$ time. Since at most $k \in \Theta(n)$ partitioning steps need to be performed, the entire partitioning strategy is of order $\mathcal{O}_{\tau, \epsilon_\tau}(n)$. Furthermore, Weight Algorithm as well as Uncoded Algorithm are of complexity order $\mathcal{O}(n)$. Therefore, the overall complexity of encoding and decoding algorithms is of order $\mathcal{O}_{\tau, \epsilon_\tau}(n)$ which completes the proof. 
\end{proof}

\section{Upper Bound on $R(\tau)$}\label{ss::upper bounds}
In this section we establish an upper bound on the rate $R(\tau)$. This upper bound is close to our lower bound for small values of $\tau$.
We make use of an approach similar to the one in~\cite{spencer2003halflie}. We take an encoding strategy and consider only messages $m\in \M$ such that any output sequence in $\mathcal{Y}_{t}^n(m)$ has a relatively large Hamming weight. For those messages, it is possible to derive a good lower bound on the size of $\mathcal{Y}_{t}^n(m)$. The upper bound on the set of possible messages is then obtained by a sphere-packing argument.
\begin{theorem}\label{th::non-trivial upper bound}
For any $\tau$, $0<\tau < 1$, we have
$$
    R(\tau)\le \overline{R}(\tau):= \min_{0\le \tau'\le \tau}  \max\limits_{\substack{0\le r \le 1, \\ h(v) \le 1 - v h\left(\min\left(\frac{\tau-\tau'}{v(1-\tau')}, \frac{1}{2}\right)\right) }} r ,
$$
where $v=v(r,\tau')$ is a real number such that $0\le v\le 1/2$ and $h(v)(1-\tau')=r$.
\end{theorem}
\begin{proof}We fix $\tau$ and $\tau'$ fulfilling the inequalities $0\leq \tau'<\tau \leq 1$ and define $t\defeq \tau n$ and $t'\defeq \tau'n$. Denote $R(\tau)$ by $\overline{r}$. Next we fix some $\epsilon>0$. We define $\overline{v} \in [0,1/2]$ as the unique real number that satisfies $h(\overline{v})(1-\tau')= \overline{r}-\epsilon$.
We define the set of output sequences of the encoding strategy when the encoder would like to transmit the message $m$ and the channel output is zero for the first $t'$ symbols to be 
$$\mathcal{Y}_{t,t'}^n(m):=\{y^n \in\mathcal{Y}_{t}^n(m):\ y_i=0\text{ for }i\in[t']\} \enspace .
$$
 For any real $v$ with $0\le v\le 1$, let $W(n,t', v)$ denote the set of all binary words $x^n$ that have $x_i=0$ for all $i\le t'$ and the Hamming weight at most $v(n-t')$. 
For $n\to\infty$, we have that the cardinality of $W(n,t',\overline{v}-\epsilon)$ is 
$$
| W(n, t',\overline{v}-\epsilon)|=\sum_{i=0}^{(\overline{v}-\epsilon)(n-t')} \binom{n-t'}{i}\le 2^{(n-t')(h(\overline{v}-\epsilon) +o(1))},
$$
where we make use of the inequality~\eqref{eq::binomial coefficient}.   Thus, there is a large enough $n_0$ so that $
|W(n, t',\overline{v}-\epsilon)|$ is at most 
$2^{(n-t')h(\overline{v})-1}$ for any $n\ge n_0$. By Definition~\ref{def::asymptotic rate}, there exists a sufficiently large integer $n>n_0$ such that  we have an encoding function~\eqref{def::encoding algorithm} for a set of messages $\M$ with $|\M|\ge 2^{n(\overline{r}-\epsilon)}$. For simplicity of notation, we assume that $(n-t')(\overline{v}-\epsilon)$ is an integer and equal to $n'$. Define the set of \textit{good} messages, written as $\M_{good}$, that consists of $m\in\M$ such that the Hamming weight of any $y^{n}\in \mathcal{Y}_{t,t'}^{n}(m)$ is at least $n'$. Since $n\ge n_0$, we obtain that $|\M_{good}|\ge |\M|- 2^{(n-t')h(\overline{v})-1}\ge 2^{(n-t')h(\overline{v})-1}$, where we used the fact $h(\overline{v})(1-\tau')=\overline{r}-\epsilon$. Now we prove that for any message $m\in \M_{good}$, the size of $\mathcal{Y}_{t,t'}^{n}(m)$ is uniformly bounded from below as follows
$$
|\mathcal{Y}_{t,t'}^{n}(m)|\ge \max\limits_{0\le \hat t\le \min(t-t', n')}\binom{n'}{\hat t}.
$$
Let $\binom{[a]}{b}$ denote the set of all possible subsets of $[a]$ of size $b$. To show the above inequality, take an arbitrary $\hat t$ with  $0\le \hat t\le \min(t-t',n')$ and define the mapping $\phi:\,\binom{[n']}{\hat t}\to \mathcal{Y}_{t,t'}^{n}(m)$ that takes an arbitrary subset $\{i_1,\ldots,i_{\hat t}\}\in \binom{[n']}{\hat t}$ with $1\le i_1 < i_2 < \ldots <i_{\hat t}\le n'$ and outputs $y^{n}\in\{0,1\}^{n}$ defined as
$$
y_i: = \begin{cases}0\quad &\text{for }i\in[t'],\\
c_{i}(m,y^{i-1})\quad &\text{for }i\in J,\\
1 - c_{i}(m,y^{i-1})\quad&\text{o/w},
\end{cases}
$$
where $J:=\bigcup\limits_{k=0}^{\hat t} [j_k+1,j_{k+1}-1]$, $j_0:=t'$, $j_{\hat t+1}:=n+1$ and for $k\in[\hat t]$, $j_k$ is the smallest $j$ so that the Hamming weight $w_H(y^{j-1}, c_{j}(m,y^{j-1})) = i_k$. One can easily see that this $y^n$ belongs to $\mathcal{Y}^{n}_{t,t'}(m)$ and for distinct $\{i_1,\ldots,i_{\hat t}\}\neq \{s_1,\ldots,s_{\hat t}\}$, the outputs $\phi(\{i_1,\ldots,i_{\hat t}\})$ and $\phi(\{s_1,\ldots,s_{\hat t})$ are different.  As the sets of output sequences are mutually disjoint, we conclude with
$$
|\M_{good}|\max\limits_{0\le \hat t\le \min(t-t', n')}\binom{n'}{\hat t}\le 2^{n-t'}.
$$
As $n$ can be taken arbitrary large, letting $n\to\infty$ yields
$$
(n-t')h(\overline{v}) + n' h\left(\min\left(\frac{t-t'}{n'},\frac{1}{2}\right)\right) + o(n) \le  n-t'.
$$
Recall that $n' = (n-t')(\overline{v}-\epsilon)$. Since the above inequality is true for any $\epsilon>0$, we have
$$
h(\overline{v}) \le 1 - \overline{v} h\left(\min\left(\frac{\tau-\tau'}{\overline{v}(1-\tau')}, \frac{1}{2}\right)\right).
$$
\end{proof}
\section{Conclusion}\label{ss::conclusion}
In this paper, we discussed a new family of error-correcting codes for the Z-channel with noiseless feedback in the combinatorial setting. By providing an explicit construction, we showed that the maximum asymptotic rate $R(\tau)$ is positive for any $\tau<1$. We have shown encoding and decoding algorithms with complexity of order $\mathcal{O}(n)$ achieving our lower bound arbitrarily close. We conjecture that the lower bound on ${R}(\tau)$ presented in Theorem~\ref{th::optimal rate} to be tight for all $\tau$. We considered feedback encoding for the case of noiseless instantaneous feedback within the encoding of every symbol. Considerations about limiting the utilization of the feedback may be an interesting starting point for subsequent research on the proposed error model. Another natural question to be asked is whether the Z-channel capacity (probabilistic setting) can be achieved by a similar encoding algorithm.
	\bibliographystyle{IEEEtran}
	\bibliography{mybibliography.bib}

\end{document}